\documentclass[]{imsart}

\RequirePackage[OT1]{fontenc}
\usepackage{amsthm,amsmath,natbib}
\RequirePackage[colorlinks,citecolor=blue,urlcolor=blue]{hyperref}

\usepackage{amssymb}
\usepackage{natbib}
\usepackage{color}
\usepackage{lmodern}
\usepackage{pdfpages}
\startlocaldefs
\newcommand*{\tr}{^\mathrm{\scriptscriptstyle T}}

\newcommand{\E}{E}
\newcommand{\Cov}{\operatorname{cov}}
\newcommand{\diag}{\operatorname{diag}}

\newcommand{\vSigma}{\boldsymbol\Sigma}
\newcommand{\vare}{\boldsymbol\epsilon}
\newtheorem{theorem}{Theorem}

\newtheorem{remark}{Remark}
\newtheorem{lemma}{Lemma}

\endlocaldefs

\begin{document}

\begin{frontmatter}

\title{Sequential estimation for GEE with adaptive variables and subject selection}

\runtitle{Sequential estimation for GEE with adaptive variables and subject selection}

\begin{aug}
\author{\fnms{Zimu} \snm{Chen}\thanksref{a,e1}\ead[label=e1,mark]{zmchen@mail.ustc.edu.cn}}
\author{\fnms{Zhanfeng} \snm{Wang}\thanksref{a,e2}\ead[label=e2,mark]{zfw@ustc.edu.cn}}
\and
\author{\fnms{Yuan-chin Ivan} \snm{Chang}\thanksref{b}%
\ead[label=e3]{ycchang@sinica.edu.tw}%
% \ead[label=u1,url]{www.foo.com}
}

\address[a]{Department of Statistics and Finance, University of Science and Technology of China\\ Hefei 230026, China
\printead{e1,e2}}

\address[b]{Institute of Statistical Science, Academia Sinica\\ Taipei 11529, Taiwan
\printead{e3}%,
% \printead{u1}
}

\runauthor{Z. Chen, Z. Wang \and YCI Chang}

% \affiliation{Some University and Another University}
\end{aug}

\begin{abstract}
Modeling correlated or highly stratified multiple-response data becomes a common data analysis task due to modern data monitoring facilities and methods. Generalized estimating equations (GEE) is one of the popular statistical methods for analyzing this kind of data.
In this paper, we present a sequential estimation procedure for obtaining GEE-based estimates. In addition to the conventional random sampling, the proposed method features adaptive subject recruiting and variable selection.  Moreover, we equip our method with an adaptive shrinkage property so that it can decide the effective variables during the estimation procedure and build a confidence set with a pre-specified precision for the corresponding parameters.  In addition to the statistical properties of the proposed procedure, we assess our method using both simulated data and real data sets.
\end{abstract}
\begin{keyword}
\kwd{Adaptive design}
\kwd{Adaptive sampling}
\kwd{Maximum quasi-likelihood estimate}
\kwd{Generalized estimating equations}
\kwd{Stopping time}
\end{keyword}

% history:
% \received{\smonth{1} \syear{0000}}

%\tableofcontents

\end{frontmatter}

\section{Introduction}
Correlated or highly stratified response data are common in studies where subjects are observed at multiple time points~\citep{diggle02}, such as in some medical, epidemiological, and financial studies and other likewise longitudinal studies.
\cite{liang1986longitudinal} proposed the generalized estimating equations (GEE) method, as an extension of generalized linear models~\citep{wedderburn1974quasi}, for analyzing this kind of correlated responses.
Due to modern computational/informational techniques, it is now easier in those scenarios to collect a long list of variables with highly clustered observations.
However, from both a theoretical and computational perspective, statisticians are still hard-pressed to efficiently decide the effective variables when building a model using GEE methods.

In addition to variable determination issues, there are usually many redundant data points in a modern, automatically collected, disorganized data set, which becomes a data analysis burden, especially  when the number of such data points is large and only limited computational capacity is available.
Using a smaller, representative subset from a large data set would be an easy way to overcome such a computational obstacle.  Nevertheless, to achieve this, we need to address new challenges, such as subset size and how to obtain a representative subset from such a data set through a statistically legitimate way. This situation becomes more complicated when there is a lengthy list of variables in the data set.
Thus, it is important to have an efficient way to simultaneously determine important variables and detect the informative data points for analyzing this kind of data.
Moreover, model interpretation would suffer if a large number of variables are included in the model. Thus, how to decide the effective variables to use in a model is essential for improving our ability to interpret it.

To balance the statistical inference needs and computational cost due to data set size, our goal is to use a sufficient amount of data to build a statistically interpretable model with effective selected variables.
  Sequential methods can help to identify and select the informative data points from a data pool for our analysis goal without using all of them at a time. This motivated us to study sequential methods with both adaptive variable and sample selection features for GEE-based estimation problems in highly stratified multiple response data.

Sequential methods are commonly used when there is no fixed sample size solution\citep[see][]{siegmund1985sequential}.
The idea of sequential methods emerged in the 1940s \citep{wald1945} and many researchers apply this concept, under different statistical setups, to many kinds of applications, such as clinical trials in medical studies, quality control in industry applications, computerized educational testing, and so on \citep{ lai2001, chang-lu2010, bartroff2013, wainer2014, tartakovsk2015, parkchang2016}.

Conventional sequential methods analyze data when they become available so that we can end an experiment as soon as we have a satisfactory  result with the desired statistical properties.  We adopt the idea of sequential analysis, but recruit data from an existing, large data pool instead of collecting new observations, which differs from conventional sequential analysis.
For a given precision in the parameter estimates, we sequentially find the most informative data points from an existing data set, for a GEE procedure, and  simultaneously select a high-impact subset of variables during the estimation process. This kind of method allows us to retrench the sample size used without diminishing study quality.
Taking advantage of the ability of sequential analysis to deal with adaptive/random subject recruiting and using a statistical experimental design  criterion,
we adaptively recruit new subjects, into the analysis to diminish the computational obstacles caused by large sample sizes.  This adaptive feature of our procedure makes it closely related to stochastic regression, sequential experimental design, and nonlinear optimal design \citep{lai-wei82, Wu1985efficient, Fedorve-Leonov2014}..

In the rest of this paper, we first briefly review the method of obtaining a maximum quasi-likelihood estimate (MQLE) and then present a sequential estimation procedure based on it for correlated multiple response data. This procedure features adaptive data selection and impact variable detection and builds a fixed-size confidence set for these effective variables when the data recruiting procedure is stopped.  In Section 3, we use both simulated data and real examples to illustrate the proposed methods for various models and under different estimation strategies.  This is then followed by the Conclusion section.  We present technical details and supplemental numerical results in the Appendix and Supplementary materials.

\section{Methodology}

\subsection{Model}
Let $(y_{ij}, \mathbf{x}_{ij})$ denote a pair representing the $j$th measurement on the $i$th subject, $j=1,\ldots, m_i$ and $i=1,\ldots, n$, where $y_{ij}$ is a scalar response and $\mathbf{x}_{ij}$ is a $p \times 1$ covariate vector.
Observations  from the same subject are assumed to be correlated; otherwise, they are independent.
{Let $\mathbf{y}_i=(y_{i1}, \ldots, y_{im_i})\tr$, $i=1, \ldots, n_i$, be the vector of responses for the $i$th cluster and $\mathbf{X}_{i}=(\mathbf{x}_{i1}, \ldots, \mathbf{x}_{im_i})\tr$ be the associated $m_i\times p$ matrix of covariates.  We assume $m_i=m<\infty$ for simplicity. }
{Let $\mathbf{h}_i(\beta)=\mathbf{h}(\mathbf{X}_i\beta)=\E\{\mathbf{y}_i\mid\mathbf{X}_i,\beta\}$} be the conditional expectation of
$\mathbf{y}_i$ given $\mathbf{X}_i$ and $\beta$, where $\mathbf{h(\cdot)}$ is an $m$-dimensional sufficiently smooth one-to-one real-valued function and $\beta$ is the unknown parameter vector to be estimated. \citet{wedderburn1974quasi} proposed an estimate  for $\beta$, {\color{black} denoted by $\tilde{\beta}_n$,} through solving the following equations:
{\color{black}
\begin{equation}\label{eq:MQLE}
  \mathbf{S}_n(\beta)=\sum_{i=1}^n\mathbf{X}_i\tr\mathbf{A}_i(\beta)\mathbf{V}_i^{-1}(\beta)\{\mathbf{y}_i-\mathbf{h}_i(\beta)\bigr\}=0,
\end{equation}
where $\mathbf{A}_i(\beta) = \mathbf{A}(\mathbf{X}_i\beta)=\partial \mathbf{h}\tr(t)/\partial t|_{t=\mathbf{X}_i\tr\beta}$ and $\mathbf{V}_i(\beta)=\mathbf{V}(\mathbf{X}_i\tr\beta)$ is a suitably-chosen known function of $\Cov\{\mathbf{y}_i|\mathbf{X}_i, \beta\}$. $\tilde{\beta}_n$ is commonly  called the MQLE.}
 If we specify the covariance of $\mathbf{y}_i$ as $\mathbf{V}_i(\beta, \alpha)=\mathbf{A}_i^{1/2}(\beta)\mathbf{R}_i(\alpha)\mathbf{A}_i^{1/2}(\beta)$  {\color{black}through a working matrix} $\mathbf{R}_i(\alpha)$, %~\citep[see][]{liang1986longitudinal},
 then \eqref{eq:MQLE} is a generalized estimating equation (GEE), as the one in \citet{liang1986longitudinal}, with equal-sized clusters.   Note that the $\tilde{\beta}_n$ derived  from  \eqref{eq:MQLE} requires no distributional assumption.

 Let $\vare_i(\beta)=\mathbf{y}_i-\mathbf{h}_i(\beta)$.  Then, following notations similar to those used in \cite{xie2003} and~\cite{Balan2005}, let
\[\mathbf{H}_n(\beta)= \sum_{i=1}^n \mathbf{X}_i\tr\mathbf{A}_i(\beta)\mathbf{V}_i^{-1}(\beta)\mathbf{A}_i\tr(\beta)\mathbf{X}_i.\]
  Then,  replacing the true (unknown)  correlation matrix {of the response variable} %$\vSigma$ 
  by $\vare_i(\beta)\vare_i\tr(\beta)$ in the covariance matrix of $\mathbf{S}_n(\beta)$, we have
\[ \mathbf{M}_n(\beta)= \sum_{i=1}^n \mathbf{X}_i\tr\mathbf{A}_i(\beta)\mathbf{V}_i^{-1}(\beta)\vare_i(\beta)\vare_i\tr(\beta)\mathbf{V}_i^{-1}(\beta)\mathbf{A}_i\tr(\beta)\mathbf{X}_i.\]
Let $\beta_0 \in R^p$ be the unknown true parameter vector to be estimated.
(To simplify  our notations, we will suppress $\beta_0$ (or $\widehat{\beta}_n$)  in the sequel when it is clear that a term is a function of $\beta$ and evaluated at $\beta_0$ (or $\widehat{\beta}_n$),  that is, $\mathbf{H}_n = \mathbf{H}_n(\beta_0)$ and $\mathbf{M}_n = \mathbf{M}_n(\beta_0)$ or, likewise, $\widehat{\mathbf{H}}_n = \widehat{\mathbf{H}}_n(\widehat{\beta}_n)$ and $\widehat{\mathbf{M}}_n = \widehat{\mathbf{M}}_n(\widehat{\beta}_n)$.)

\subsection{Sequential method for MQLE}
Using the notations defined before, let $\mathcal{F}_{n} = \sigma\{(\mathbf{y}_j, \mathbf{X}_j) : j = 1, \ldots ,n\}$ for $n \geq 1$ and
let $\mathcal{F}_{0} = \sigma\{\emptyset, \Omega\}$;  then, $\{ \mathcal{F}_{n}:  n \geq 0 \}$ is an increasing sequence of $\sigma$-fields.
Assume that for each $i$, the observed data satisfy {\color{black}$\E\{\mathbf{y}_i \mid \mathcal{F}_{i-1}\} = \mathbf{h}_i(\beta_0)$}.
Define $\mathbf{e}_i = \mathbf{y}_i - \E\{\mathbf{y}_i \mid \mathcal{F}_{i-1}\}$ for each $i \geq 1$;  then, $\{\mathbf{e}_n\}$ is a martingale difference sequence with respect to the $\sigma$-fields $\mathcal{F}_{n}$.
{Following \cite{lai-wei82}, we also call Equation~\eqref{eq:MQLE} a stochastic quasi-likelihood estimating equation.}

{
If there are only a small number $p_0$ out of the $p$ components of $\beta_0$ that truly have an impact on the response and we know which they are in advance, then fitting a model with these $p_0$ variables, without using all $p$ variables, will certainly be more efficient -- requiring fewer observations and less computational time.
In practice, we usually have no such information, nor the number $p_0$,  and fitting a model with all of the variables would increase the computational cost and instability of the model.
The situation will be worse when $p$ is much larger than $p_0$. (Throughout the rest of this paper, we will refer to the $p_0$ variables as `effective variables’.)}   The conventional sequential estimation procedure, with no ability to detect the effective variables, cannot work well in this situation, which motivates us to incorporate the idea of adaptive shrinkage estimate (ASE) \citep{wang2013sequential} to the current procedure.

{\color{black}
There is a great deal of discussion about the asymptotic properties of the estimate for stochastic linear regression. Here, without linearity assumptions, we report some asymptotic results for MQLEs under conditions similar to those used in \citet{Lai1979StrongCon, lai-wei82} for adaptive covariate/design vectors:}
\renewcommand\labelenumi{(C\arabic{enumi})}
\renewcommand\theenumi\labelenumi
\begin{enumerate}
\item For any $t \in R^q$, $\mathbf{V}(t) > 0$, $\det \mathbf{A}(t) \neq 0$, each element of $\mathbf{V}(t)$ is continuously differentiable, and {\color{black}$\mathbf{h}(t)$} is twice continuously differentiable.
\item The matrix $ \mathbf{X}_i$, $i \geq 1$, satisfies $\sup_{i\geq1} \|\mathbf{X}_i\| < \infty$ almost surely, {\color{black}where $\|\cdot\|$ stands for the Euclidean norm}.  In addition, the maximum and minimum eigenvalues of $\sum_{i=1}^{n}\mathbf{X}_i\mathbf{X}_i\tr$,  denoted  respectively as $\bar{\lambda}_n$ and $\underline{\lambda}_n$,  satisfy $\underline{\lambda}_n \rightarrow \infty$ and
${\color{black}\lim\inf_{n\rightarrow\infty} \underline{\lambda}_n/\{(\bar{\lambda}_n\log\bar{\lambda}_n)^{\frac 1 2} (\log\log\bar{\lambda}_n)^{\frac 1 2 + \alpha}\} > 0}$ almost surely for some  $\alpha > 0$.
\item The martingale difference sequence $\{\mathbf e_n, n\geq 1\}$ with respect to an increasing sequence of $\sigma$-fields $\{\mathcal{F}_{n}, n\geq 1\}$ satisfies
\[\sup_{i\geq1} \E\{\|\mathbf{e}_i\|^2\mid\mathcal{F}_{i-1}\} < \infty \ \text{almost surely}.\]
\item $\Cov\{\mathbf{e}_i\mid\mathcal{F}_{i-1}\} > c{I}_q$ for all $i\geq1$, and $\sup_{i\geq1} \E\{\|\mathbf{e}_i\|^r\mid\mathcal{F}_{i-1}\} < \infty$ almost surely for some $r > 2$, {\color{black}where $c$ is a positive constant independent of $n$} and $I_q$ is a $q\times q$ identity matrix.
\item There exists a non-random positive definite symmetric matrix $\tilde{\mathbf{M}}_n$, for which \\
$\tilde{\mathbf{M}}_n^{-1/2}\mathbf{M}_n\tilde{\mathbf{M}}_n^{-1/2}\rightarrow {I}_p$ in probability,
where
$\tilde{\mathbf{M}}_n^{-1/2}$ is the positive definite symmetric square root of $\tilde{\mathbf{M}}_n$.
\item There exists a continuously increasing function $\rho(\cdot)$
and a non-random positive definite symmetric matrix $\vSigma$ such that
\begin{align}\label{eq:cond.asynorm}
   \mathbf{M}_n^{-1/2}\mathbf{H}_n/{\rho(n)}^{1/2}\rightarrow \vSigma^{1/2} \text{ almost surely.}
\end{align}
\end{enumerate}

Assume that conditions (C1)--(C5) are satisfied. Then, we have
\begin{lemma}\label{mqle:ucip}
The sequence of random variable $\bigl\{\widehat{\mathbf{M}}_{n}^{-1/2}\widehat{\mathbf{H}}_{n}(\widehat \beta _n-\beta_0),n=1,2,\ldots\bigr\}$ is uniformly continuous in probability {\color{black} (u.c.i.p.)}.
\end{lemma}

\subsection{Adaptive shrinkage estimate}

{
Under Conditions (C1)--(C3), \citet{yin2006asymptotic, Yin2008strongrate} showed that the convergence rate of the MQLE $\tilde{\beta}_n$ is:
\begin{align}\label{eq:con-rate}
\|\tilde\beta_n - \beta_0\| =
o(\{(\overline\lambda_{n}\log \overline\lambda_{n})^{1/2}(\log\log \overline\lambda_{n})^{1/2+\alpha}\}/\underline\lambda_{n}).
\end{align}
Hence, if we let $L_r = \{(\overline\lambda_{n}\log \overline\lambda_{n})^{1/2}(\log\log \overline\lambda_{n})^{1/2+\alpha}\}/\underline\lambda_{n}$, then, as $n\rightarrow \infty$, with probability one,
{${L_r}^{1/2}\kappa|\tilde\beta_{nj}|^{-\gamma}\longrightarrow 0\times I(\beta_{0j}\neq 0)+\infty\times
  I(\beta_{0j}=0)$,
where $I(\cdot)$ is the indicator function (presuming $\infty\times 0=0$) and \(\kappa=\kappa(n)\) is a non-random function of \(n\) such that
for $ 0<\delta<\frac{1}{2} $ and $\gamma>0$,
\begin{equation}\label{eq:kappa.rate}
{L_r}^{1/2}\kappa\rightarrow 0\quad \text{and} \quad
{L_r}^{1/2+\gamma\delta}\kappa\rightarrow\infty,\quad \text{as} \quad n\rightarrow\infty.
\end{equation}
(Note that~\eqref{eq:kappa.rate} can be easily satisfied; for example, we can take $\kappa= {L_r}^{-\theta}$ with $\theta\in(\frac{1}{2},\frac{1}{2}+\gamma\delta)$.)
Let $I_{nj}(\epsilon)=I\{{L_r}^{1/2} \kappa|\tilde\beta_{nj}|^{-\gamma}<\epsilon\}$  for each $j$; then, by \eqref{eq:con-rate}, we can use $I_{nj}(\epsilon)$ to detect whether the absolute difference between the {$j$th} component of $\beta_{0}$ and zero is asymptotically significantly larger than a predetermined constant $\epsilon ( >0)$.
Let $I_n(\epsilon)=\diag\{I_{n1}(\epsilon),\ldots, I_{np}(\epsilon)\}$; then,  $\widehat\beta_n \equiv I_n(\epsilon)\tilde\beta_n$ defines an ASE of $\beta_0$.
Thus, we have the following theorem:
\begin{theorem}\label{thm:ucip.rho}
Let  $N(t)$ be a positive integer-valued random variable for which $N(t)/t$ converges to $1$ in probability as $t\rightarrow\infty$.
{Assume that conditions in Lemma~\ref{mqle:ucip} and (C6) are satisfied},
then for any small $\epsilon > 0$, as $n\rightarrow\infty$ we have
 $\sqrt{\rho(N(t))}(\widehat\beta_{N(t)}-\beta_0)\rightarrow N(0,I_0\vSigma^{-1}I_0)$ in distribution,  where $I_0=\diag\bigl\{I(\beta_{01}\neq 0), \ldots, I(\beta_{0p}\neq 0)\bigr\}$ is a $p\times p$ diagonal matrix.
\end{theorem}
}

For convenience, we rearrange the matrix according to the effective variables detected at the current stage.
Let $O_n$ be an orthonormal matrix satisfying $O_n\tr O_n=I_p$ so that $(\widehat{\beta}_{n1}\tr ,\widehat{\beta}_{n2}\tr )\tr =O_n\widehat\beta_n$, where $(\widehat{\beta}_{n1}\tr ,\widehat{\beta}_{n2}\tr )\tr $ are the order-rearranged components of $\widehat\beta_n$ in which the indicators ${I_{nj}}$ corresponding to $\widehat\beta_{n1}$ are equal to 1, and the remaining ones, corresponding to $\widehat\beta_{n2}$, are equal to 0.
We partition the matrix $(O_n\widehat{\mathbf{H}}_n\widehat{\mathbf{M}}_n^{-1/2})(O_n\widehat{\mathbf{H}}_n\widehat{\mathbf{M}}_n^{-1/2})\tr $ according to the first $\hat p_0$ nonzero components of $O_n\widehat\beta_n$ so that
{\color{black}
\begin{equation}\label{eq.1}
(O_n\widehat{\mathbf{H}}_n\widehat{\mathbf{M}}_n^{-1/2})(O_n\widehat{\mathbf{H}}_n\widehat{\mathbf{M}}_n^{-1/2})\tr =\left(
\begin{array}{cc}
{\Sigma_{11}(n)}_{\hat p_0\times\hat p_0} & {\Sigma_{12}(n)}_{\hat p_0\times(p_0-\hat p_0)} \\
{\Sigma_{21}(n)}_{(p-\hat p_0)\times\hat p_0}&{\Sigma_{22}(n)}_{(p-\hat p_0)\times (p-\hat p_0)}\\
\end{array}\right).
\end{equation}
With simple matrix algebra, we have that
\begin{align}\label{sigma-1}
&\;O_nI_n(\epsilon)\bigl((\widehat{\mathbf{H}}_n\widehat{\mathbf{M}}_n^{-1/2})(\widehat{\mathbf{H}}_n\widehat{\mathbf{M}}_n^{-1/2})\tr \bigr)^{-1}
I_n(\epsilon){O_n}\tr \nonumber\\
=&\;O_nI_n(\epsilon){O_n}\tr \bigl((O_n\widehat{\mathbf{H}}_n\widehat{\mathbf{M}}_n^{-1/2})(O_n\widehat{\mathbf{H}}_n\widehat{\mathbf{M}}_n^{-1/2})\tr \bigr)^{-1}
O_nI_n(\epsilon){O_n}\tr \nonumber\\
=&\;\left(
\begin{array}{cc}
{\tilde{\Sigma}_{11}}^{-1}(n) & 0 \\
0 & 0 \\
\end{array}\right),
\end{align}
where
$\tilde{\Sigma}_{11}^{-1}(n)={\Sigma^{-1}_{11}}(n)
+{\Sigma^{-1}_{11}}(n)\Sigma_{12}(n){\Sigma^{-1}_{22.1}}(n)\Sigma_{21}(n){\Sigma^{-1}_{11}}(n)$
and
${\Sigma^{-1}_{22.1}}(n)=\Sigma_{22}(n)-\Sigma_{21}(n){\Sigma^{-1}_{11}}(n)$ $\Sigma_{12}(n).$
Let $M^{-}$ denote a general inverse matrix $M$. It follows that
\begin{align*}
&\;(\widehat\beta_n-\beta_0)\tr \bigl(I_n(\epsilon)((\widehat{\mathbf{H}}_n\widehat{\mathbf{M}}_n^{-1/2})(\widehat{\mathbf{H}}_n\widehat{\mathbf{M}}_n^{-1/2})\tr )^{-1}I_n(\epsilon)\bigr)^{-}(\widehat\beta_n-\beta_0)\nonumber\\
=&\;(\widehat\beta_{n1}-\beta_{01})\tr \tilde{\Sigma}_{11}(n)(\widehat\beta_{n1}-\beta_{01}),
\end{align*}
where $\beta_{01}$ is sub-vector of $\beta_0$ corresponding to $\widehat\beta_{n1}$.
Using Theorem~\ref{thm:ucip.rho}, it implies that as $t\rightarrow\infty$,
\begin{equation}\label{chi2}
  (\widehat\beta_{N1}-\beta_{01})\tr \tilde{\Sigma}_{11}(N)(\widehat\beta_{N1}-\beta_{01})\rightarrow{\chi}^2(p_0) \quad \text{in distribution.}
\end{equation}

\subsection{Sequential estimate with adaptive shrinkage estimate}

 Our goal now is to have  a sequential estimation procedure that can detect the effective variables and then construct a fixed-sized confidence set for them with a pre-specified accuracy.  Theorem~\ref{thm:ucip.rho} paves the way for models with stochastic regressors.
However,  due to the unknown true value of $p_0$, which should also be estimated using the observations,   we cannot use \eqref{chi2} directly.
Conventional sequential estimation focuses on the sample size required for the parameter estimates; here, we need to decide the effective variables sequentially based on the recruited observations as well.

Suppose that $C_k$ is a set of $k$ observations $\{(\mathbf{y}_i,\mathbf{X}_i):i=1,\ldots,k\}$ at the $k$th stage.
For an $\epsilon > 0$, let $\hat p_0(k)  =  \sum_{j=1}^p I_{kj}(\epsilon)$ based on $C_k$.
Because $I_{nj}(\epsilon)$ converges to $I(\beta_{0j}\neq 0)$ almost surely as $n\rightarrow\infty$, this implies that $\hat{p}_0 = \hat{p}_0(n)$ converges to $p_0$ almost surely; in addition,  we have that $\lim_{n\rightarrow \infty}\E\{\hat{p}_0(n)\} = p_0$ \citep[see also][]{wang2013sequential}.

Let $a_k^2\in R$ be a constant satisfying the conditional probability $\mathrm{pr}(\chi^2_{\hat p_0(k)}\leq a_k^2|C_k)=1-\alpha$ for a given $\alpha$.
For a given $d > 0$, define the stopping time $N_d$ as follows:
\begin{eqnarray}\label{eq:stoprule.noshrinkage}
  N=N_d= \inf \left\{k: k\geq n_0\  {\text{and}}\  \nu_k\leq \frac{d^2}{a_k^2}\right\},
\end{eqnarray}
where $\nu_k$ is the maximum eigenvalue of
$I_k(\epsilon)(\widehat{\mathbf{H}}_k\widehat{\mathbf{M}}_k^{-1}\widehat{\mathbf{H}}_k)^{-}I_k(\epsilon)$.
We then conduct a sequential estimation procedure, collecting one new observation at a time until the stopping criterion $N_d$ is satisfied.
If the inequality in \eqref{eq:stoprule.noshrinkage} is fulfilled, then using all the $N_d$ observations, we construct  a confidence ellipsoid for $\beta_0$:
\begin{align}\label{cset}
  R_N=\left\{\beta\in R^p: S_N\leq \frac{d^2}{\nu_N}\  {\mbox{and}}\  \beta_j=0\  {\mbox{for}}\  I_{Nj}(\epsilon)=0,1\leq j\leq p\right\},
\end{align}
where $S_N=(\beta_{N1}-\widehat\beta_{N1})\tr \tilde{\Sigma}_{11}(N)(\beta_{N1}-\widehat\beta_{N1})$.
It follows that the maximum axis of ${(\beta_{N1}-\widehat\beta_{N1})}\tr {\tilde\Sigma}_{11}(N){(\beta_{N1}-\widehat\beta_{N1})}={d^2}/{\nu_N}$, say $R_N^{\hat p_0}$, is $2d$.  Please note that this $R_N^{\hat p_0}$ is only for the  $\hat p_0$ detected effective variables, which are sequentially decided based on the observations.

The proposed method allows the sequential estimation procedure to focus on the effective variables so that the confidence ellipsoid $R_N^{\hat p_0}$ is a projection of $R_N$ onto the $\hat p_0$-dimensional space spanned by the axes with nonzero components of $\widehat\beta_N$.  This is the reason why the proposed method uses fewer observations to achieve the required properties compared to those used in procedures with no variable detection feature.
Theorem \ref{thm:shrinkage.seq} summarises these properties; its proof is in the Appendix.
\begin{theorem}\label{thm:shrinkage.seq}
Let $N$ be the stopping time as defined in \eqref{eq:stoprule.noshrinkage}. {\color{black} If the assumptions of Theorem~\ref{thm:ucip.rho} hold, then}
(i) $\lim_{d\rightarrow 0}{d^2{\color{black}\rho(N)}}/{a^2\nu}=1$ almost surely,
(ii) $\lim_{d\rightarrow 0}\mathrm{pr}(\beta_0\in R_N)=1-\alpha$,
(iii) $\lim_{d\rightarrow 0}\hat p_0(N)=p_0$ almost surely, and
(iv) $\lim_{d\rightarrow 0}\E\{\hat p_0(N)\}=p_0$,
where $\nu$ is the maximum eigenvalue of matrix {\color{black}$I_0\vSigma^{-1}I_0$}.
\end{theorem}
Theorem \ref{thm:shrinkage.seq} parts (i) and (ii)  were termed `asymptotic consistency and efficiency,' respectively \citep{chow1965asymptotic}, which mean that (i) the coverage probability of the proposed sequential procedure will converge to the prescribed one and (ii) the proposed sequential procedure is, asymptotically, as efficient as the (unknown) fixed sample procedure in terms of the ratio of the sample sizes.  Theorem \ref{thm:shrinkage.seq} parts (iii) and (iv) state that when the sequential estimation procedure is stopped, the estimate of the number of effective variables converges to the true number of effective variables in probability and its expectation will also converge to the true number of effective variables.

By \eqref{eq:stoprule.noshrinkage}, it is clear that the procedure requires a larger sample size for a smaller $d$. From Equation \eqref{cset}, we can see that the constant $d$ confines the maximum axis of the confidence ellipsoid, which specifies the precision of the final parameter estimates of a sequential procedure.   Together, these guarantee  that the proposed sequential procedure can use  a `minimum' required sample size such that the confidence set will approximately have the pre-specified coverage probability for the effective components of $\beta_0$.
(Note that to sample a batch of new observations at a time is also feasible with slight modifications of the current arguments.)

\begin{remark}
We can treat the conventional  sequential procedure, without  ASE, as a simplified case of the current procedure.
Hence, we also summarize the results for the fixed-sized confidence set estimation without the variable detection feature in {Supplementary materials}. %the Appendix.
\end{remark}

\subsection*{Modified Sequential D-criterion}
In the theorem above, the regularity conditions on the design are not for any specific design scheme.
Using the D-optimality criterion in statistical experimental design to select points will fulfill those conditions.
However, to find the new observations based on conventional D-optimality may be computational intensive.
Therefore,  we use a modified sequential D-criterion described below to select new observations adaptively and sequentially. 

{\color{black}
Let {$\mathcal{D}$} and  $\mathcal{U}$ be the recruited sample set and the inactive sample set, respectively.
The regular D-criterion is to choose a data point $\mathbf{X}^*$ from the data pool, which satisfies the following equation:
$$
  \mathbf{X}^*=\arg\max_{\mathbf{X}^*\in\mathcal{U}}\det\Bigl(\mathbf{H}_{(n+1)1}(\widehat{\beta}_n)\mathbf{M}_{(n+1)1}^{-1}(\widehat{\beta}_n)\mathbf{H}_{(n+1)1}(\widehat{\beta}_n)\Bigr).
$$
{\color{black} Because this procedure is usually computationally intensive},  we propose the following modified sequential D-criterion instead:
\begin{align*}
      \mathbf{X}^*&=\arg\max_{\mathbf{X}^*\in\mathcal{U}}\det\bigl(\mathbf{G}_{(n+1)1}\bigr)
      =\arg\max_{\mathbf{X}^*\in\mathcal{U}}\det\bigl(\mathbf{G}_{(n)1}+\mathbf{g}_{(n+1)1}\bigr),
\end{align*}
where $$\mathbf{G}_n=\sum_{i=1}^n\mathbf{g}_i=\sum_{i=1}^n \mathbf{X}_i\mathbf{A}_i(\widehat{\beta}_n)\widehat{\mathbf{V}}_i^{-1}(\widehat{\beta}_n)\mathbf{A}_i\tr(\widehat{\beta}_n)\mathbf{X}_i\tr=\sum_{i=1}^n \mathbf{X}_i\mathbf{A}_i^{1/2}(\widehat{\beta}_n)\bar{\mathbf{R}}^{-1}\mathbf{A}_i^{1/2}(\widehat{\beta}_n)\mathbf{X}_i\tr$$ and
$\bar{\mathbf{R}}=\frac 1 n \sum_{i=1}^n\mathbf{A}_i^{-1/2}(\widehat{\beta}_n)\{\mathbf{y}_i-\mathbf{h}_i(\widehat{\beta}_n)\}\{\mathbf{y}_i-\mathbf{h}_i(\widehat{\beta}_n)\}\tr\mathbf{A}_i^{-1/2}(\widehat{\beta}_n)$, as suggested in \cite{Balan2005}.
{\color{black}  (For a matrix $\mathbf{W}$, the notation $\mathbf{W}_{n}$ represents the  matrix arranged according to the detected effective variables at stage $n$ and $\mathbf{W}_{(n)1}$ denotes its sub-matrix.}
After $\mathbf{X}^*$ is selected, we delete {it} from $\mathcal{U}$ and add $(\mathbf{y}^*, \mathbf{X}^*)$ to {$\mathcal{D}$}, then update $\widehat\beta_n$ to $\widehat\beta_{n+1}$. Then, we continue this process until the stopping criterion \eqref{eq:stoprule.noshrinkage} is fulfilled.
}

\section{Numerical Studies}

\subsection{Simulated Data}
We first assess the proposed method with simulated data sets by comparing the average sample sizes (stopping times, $N^*$), empirical coverage probabilities (C.P.), and consistency of variable detection of the GEE-based procedures with ASE under random sampling (ASE-R), adaptive sampling with D-optimality criterion (ASE-D), and the procedure with all $p$ variables (GEE).  Additionally, we report the results using only the true $p_0$ variables (the Oracle) as a reference.  We also report the average numbers of incorrectly identified zero variables {\color{black}($\text{Num}_{ic}$)} and the average number of correctly identified zero variables {\color{black}($\text{Num}_c$)} for the performance in identifying effective variables.  The targeted coverage probability is equal to  $95\%$ and the empirical coverage probability should approach the nominal $95\%$ as $d$ decreases.   Let $\kappa=d^2N/(a^2_N\nu_N)$; then, as stated in Theorem~\ref{thm:shrinkage.seq}, we expect that $\kappa$ will be close to 1 as $d$ decreases.  We consider both continuous and discrete responses in our simulation study.  

\subsection*{Continuous Response Case.}
{\color{black}
The highlight of this example is about the ability to detect effective variables and we generate data with a linear link function for illustration purposes.
Let $\mathbf{X}_i=(\mathbf{x}_{i1}, \ldots, \mathbf{x}_{im})$ be the covariate matrices with $m=5$ and $p= 8, 10$, and $24$,
and for different $p$, the parameter $\beta_0=(1, -1.1, 1.5, -2, \mathbf{0}_k)\tr$ with $k=4, 8$, and $20$,  where
$\mathbf{0}_k$ denotes the $k$-dimensional vector of $0$s.
Thus, the ratio of the number of nonzero coefficients ($p_0$) to the number of zero coefficients ($p_k$) are  $4:4$, $4:8$ and $4:20$.
We generate this data set sequentially  so that successors depend on the information of their predecessors in the following adaptive  manner:
Let the first $\mathbf{X}_1$ be from a standard multivariate normal distribution with zero mean vector and identity covariance matrix. We
then adaptively generate $\mathbf{X}_n$s, $n >1$, from a multivariate normal distribution with mean vector $\sum_{i=1}^{n-1}\sum_{j=1}^{m} \mathbf{x}_{ij} /m(n-1)$ and identity covariance matrix.
We then generate the response  $\mathbf{y}_i$ as follows:
\begin{equation*}
  \mathbf{y}_i = \mathbf{X}\tr_i\beta_0 + \vare_i,
\end{equation*}
where the random error vector $\vare_i$ follows a normal distribution with mean $\mathbf{0}$ and three different covariance structures with corresponding dimensional numbers.  These three covariance matrices are the identity matrix, the exchangeable, and the \textsc{ar}(1) autoregressive correlation structure with autocorrelations $\alpha=0.3$ and $0.5$.
We only report the results of the {\textsc{ar}(1)} case with {$p_0:p_k = 4:20$} and present the rest of the results in Supplementary materials.  In this example, we set the parameters for ASE, as described in \eqref{eq:kappa.rate}, as follows:  $\gamma=1$, $\delta = 0.45$, and $\theta = 0.65$  and  use the \textsc{qic} \citep{pan2001QIC} as its information criterion.
}

We start with a  data pool of $1000$ data points each time and obtain an initial estimate of $\beta_0$ using a random sample of initial data of size $25$ from this data pool. Table \ref{tab:sim1.ar1.res.4vs20.con} summarizes the results based on 1000 replications with $d=0.2, 0.3$, and $0.5$.  The results show that the empirical coverage probabilities of ASE-D, ASE-R, and GEE are all less than that of the Oracle, whereas the sample sizes used in ASE-D and ASE-R are very close to that of the Oracle. In this example,  the estimate from the conventional GEE uses around 4 or 5 times larger sample sizes than those of the other methods.   The incorrectly identified zero variables ($\text{Num}_{ic}$) are all equal to 0. For ASE-D and ASE-R, the average number of correctly identified zero variables ($\text{Num}_c$) is close to the true number -- 20 and all GEE cases have $\text{Num}_c < 0.1$, which is reasonable, since conventional GEE do not consider the sparseness situation.  These results confirm that the proposed sequential procedure with ASE works well for the multiple correlated responses data in the current simulation setup.  Both ASE-D and ASE-R perform similarly here, because this model is simple.
The advantage of ASE-D will be clearly revealed in the discrete response case next.

\begin{table}
% \def~{\hphantom{0}}
% \tabcolsep=3pt\fontsize{10}{9}
% \selectfont
\centering
\caption{Results of sequential sampling method, performance in variable identification, and estimation of number of nonzero components for identity link function with \textsc{ar}(1) correlation when $p_0:p_k=4 : 20$.}\label{tab:sim1.ar1.res.4vs20.con}

\begin{tabular}{cclrcrrr}%\toprule
\hline
    {$\alpha$} & {$d$} & Method  & \multicolumn{1}{c}{$N^*$}  & $C.P.$  & \multicolumn{1}{c}{$\kappa$}  & {$Num_{c}$} & {$Num_{ic}$} \\% \midrule
    \hline
    0.3 & 0.4 & Oracle &34.498  (4.931)   & 0.887  & 1.079  & 0.000   & 0.000  \\
        &     & ASE-D  &36.452  (6.278)   & 0.669  & 1.126  & 19.714  & 0.000\\
        &     & ASE-R  &40.947  (11.430)  & 0.604  & 1.133  & 19.552  & 0.000\\
        &     & GEE    &156.375 (12.067)  & 0.690  & 1.009  & 0.033   & 0.000\\
        & 0.3 & Oracle &55.091  (9.457)   & 0.901  & 1.021  & 0.000   & 0.000  \\
        &     & ASE-D  &52.116  (10.537)  & 0.827  & 1.038  & 19.948  & 0.000\\
        &     & ASE-R  &63.010  (17.939)  & 0.779  & 1.078  & 19.834  & 0.000\\
        &     & GEE    &245.548 (16.103)  & 0.810  & 1.006  & 0.048   & 0.000\\
        & 0.2 & Oracle &118.769 (14.745)  & 0.928  & 1.008  & 0.000   & 0.000  \\
        &     & ASE-D  &113.290 (16.222)  & 0.904  & 1.010  & 19.997  & 0.000\\
        &     & ASE-R  &120.243 (16.195)  & 0.882  & 1.012  & 19.986  & 0.000\\
        &     & GEE    &488.373 (25.909)  & 0.892  & 1.002  & 0.066   & 0.000\\
    0.7 & 0.4 & Oracle &40.677  (8.140)   & 0.880  & 1.042  & 0.000   & 0.000  \\
        &     & ASE-D  &40.920  (8.357)   & 0.769  & 1.052  & 19.981  & 0.000\\
        &     & ASE-R  &41.707  (8.640)   & 0.781  & 1.051  & 19.950  & 0.000\\
        &     & GEE    &170.199 (18.057)  & 0.695  & 1.005  & 0.053   & 0.000\\
        & 0.3 & Oracle &76.714  (12.918)  & 0.918  & 1.013  & 0.000   & 0.000  \\
        &     & ASE-D  &76.442  (12.987)  & 0.862  & 1.015  & 20.000  & 0.000\\
        &     & ASE-R  &77.342  (13.187)  & 0.882  & 1.015  & 19.999  & 0.000\\
        &     & GEE    &303.726 (23.976)  & 0.830  & 1.003  & 0.072   & 0.000\\
        & 0.2 & Oracle &172.986 (18.759)  & 0.931  & 1.005  & 0.000   & 0.000  \\
        &     & ASE-D  &173.148 (19.087)  & 0.919  & 1.005  & 20.000  & 0.000\\
        &     & ASE-R  &173.864 (19.172)  & 0.917  & 1.005  & 20.000  & 0.000\\
        &     & GEE    &676.856 (36.565)  & 0.905  & 1.001  & 0.108   & 0.000\\
     %\bottomrule
      \multicolumn{8}{p{20em}}{* Standard deviations are given in parentheses} \\
\hline
\end{tabular}%
  % \label{tab:sim1.ar1.res.4vs20.con}%
\end{table}%

\subsection*{Discrete Responses Case}
{
In this example, we use a logistic model to illustrate discrete response situations.
We generate the covariates vector $\mathbf{x}_{ij}$ from a multivariate normal distribution with mean zero and an \textsc{ar}(1) correlation matrix with autocorrelation coefficient $0.5$ and marginal variance equal to $0.2$.
The binary response vector for each cluster has an \textsc{ar}(1) correlation structure with correlation coefficient $\alpha$ and the marginal expectation $\mu_{ij}$ satisfies the following equation:
\begin{equation*}
  \text{logit}(\mu_{ij})=\mathbf{x}_{ij}\tr\beta_0, \quad i=1,\ldots, 5000,\quad j=1,2,3,
\end{equation*}
We use the \texttt{R} program %\citep{team2013r} program 
in \citet{Oman2009} to generate the correlated binary response data under several different setups:
two values of $\alpha = 0.3$ and $0.7$, and three different length regression coefficient vectors:
$\beta_0\tr=(0.6, -0.5, 0.4, \mathbf{0}_k\tr$) with $k=3, 6$ and $12$.

At the beginning, 5000 observations are generated as the data pool and we randomly choose $200$ from this pool as our initial set.
We consider two situations in this nonlinear case: (1) the structure of $\mathbf{R}(\alpha)$ %in $\mathbf{V}_i=\mathbf{A}_i^{1/2}(\beta)\mathbf{R}(\alpha)\mathbf{A}_i^{1/2}(\beta)$
is known and (2) the correlation structure $\mathbf{R}(\alpha)$ is mis-specified. %The further details can be found in Supplementary materials.
In both situations, we estimate $\alpha$ based on the suggestion of \citet{liang1986longitudinal}.

\begin{table}
% \def~{\hphantom{0}}
%  \tabcolsep=3pt\fontsize{10}{9}
% \selectfont
\centering
\caption{ Performance in variable identification and estimation of nonzero components under sequential sampling method based on ASE and GEE for a logit link function and \textsc{ar}(1) correlation between clusters when $p_0:p_k=3 : 12$.}\label{tab:sim2.ar3vs12.res.con}
\begin{tabular}{cclrcrrr}%\toprule
\hline
    {$\alpha$} & {$d$} & Method  & \multicolumn{1}{c}{$N$}  & $C.P.$  & \multicolumn{1}{c}{$\kappa$}  & {$Num_{c}$} & {$Num_{ic}$} \\% \midrule
    \hline
   0.3  & 0.5 & Oracle & 384.997  (41.880)  & 0.949 & 1.004 & 0.000  & 0.000 \\
        &     & ASE-D  & 297.593  (52.953)  & 0.873 & 1.022 & 11.675 & 0.434 \\
        &     & ASE-R  & 415.522  (114.443) & 0.858 & 1.016 & 11.618 & 0.342 \\
        &     & GEE    & 1752.497 (77.876)  & 0.937 & 1.001 & 0.014  & 0.000 \\
        & 0.4 & Oracle & 603.168  (56.069)  & 0.950 & 1.003 & 0.000  & 0.000 \\
        &     & ASE-D  & 381.430  (93.703)  & 0.923 & 1.014 & 11.872 & 0.402 \\
        &     & ASE-R  & 604.040  (167.815) & 0.908 & 1.010 & 11.873 & 0.279 \\
        &     & GEE    & 2680.709 (93.850)  & 0.938 & 1.001 & 0.015  & 0.000 \\
        & 0.3 & Oracle & 1070.389 (72.207)  & 0.956 & 1.002 & 0.000  & 0.000 \\
        &     & ASE-D  & 611.206  (175.545) & 0.943 & 1.007 & 11.973 & 0.331 \\
        &     & ASE-R  & 1090.787 (248.695) & 0.951 & 1.005 & 11.983 & 0.167 \\
        &     & GEE    & 4693.950 (123.752) & 0.950 & 1.000 & 0.021  & 0.000 \\
   0.7  & 0.5 & Oracle & 384.904  (43.847)  & 0.941 & 1.004 & 0.000  & 0.000 \\
        &     & ASE-D  & 296.432  (54.766)  & 0.902 & 1.020 & 11.660 & 0.454 \\
        &     & ASE-R  & 412.007  (108.277) & 0.855 & 1.017 & 11.615 & 0.345 \\
        &     & GEE    & 1748.791 (73.967)  & 0.948 & 1.001 & 0.012  & 0.000 \\
        & 0.4 & Oracle & 602.695  (54.798)  & 0.951 & 1.003 & 0.000  & 0.000 \\
        &     & ASE-D  & 383.120  (92.708)  & 0.916 & 1.014 & 11.879 & 0.383 \\
        &     & ASE-R  & 615.738  (162.033) & 0.903 & 1.010 & 11.861 & 0.265 \\
        &     & GEE    & 2685.861 (92.699)  & 0.955 & 1.001 & 0.015  & 0.000 \\
        & 0.3 & Oracle & 1071.102 (74.639)  & 0.953 & 1.002 & 0.000  & 0.000 \\
        &     & ASE-D  & 616.704  (173.345) & 0.921 & 1.007 & 11.969 & 0.314 \\
        &     & ASE-R  & 1087.811 (250.271) & 0.941 & 1.006 & 11.974 & 0.169 \\
        &     & GEE    & 4696.926 (129.885) & 0.961 & 1.000 & 0.019  & 0.000 \\
        \hline
    \end{tabular}%
\end{table}%

Table \ref{tab:sim2.ar3vs12.res.con} reports the results of \textsc{ar}(1) with $p_0:p_k=3 : 12$.
The GEE method using all $p$ variables has better coverage probabilities than the other two procedures at the cost of the sample size used.
In fact, GEE usually uses samples approximately 4 to 5 times larger than those used by the Oracle procedure in this study.
Both ASE-D and ASE-R have similar values of $\text{Num}_{c}$, which are close to the targeted $p_k= 12$, whereas the $\text{Num}_{ic}$ values are slightly higher than those in the previous case due to the nonlinear model situation.
In the current model, we clearly see that the sample sizes used in ASE-D and ASE-R are very different: ASE-D uses fewer observations by fully taking the advantage of the D-optimality criterion in selecting data points. The results for the other settings are in Supplementary materials.

}

\subsection{Real Data Examples}
\subsection*{Yeast Cell-Cycle Gene Expression Data Analysis} We apply the proposed MQLE-based sequential procedure with {\color{black}{ASE}} to the yeast cell-cycle gene expression data set collected in CDC15 \citep[][see also http://genome-www.stanford.edu/cellcycle/data/rawdata/]{Spellman1998}. In this experiment, genome-wide mRNA levels were recorded for 6178 yeast ORFs (an abbreviation for open reading frames, DNA sequences that can determine which amino acids will be encoded by a gene) at 7 minute intervals for 119 minutes, which covers two cell-cycle periods, for a total of 18 time points.
The cell-cycle is a tightly regulated life process, where cells grow, replicate their DNA, segregate their chromosomes, and divide into as many daughter cells as the environment allows, and this process is commonly divided into M/G1-G1-S-G2-M stages. The M stage stands for `mitosis', during which nuclear (chromosome separation) and cytoplasmic (cytokinesis) divisions occur. The G1 stage stands for `GAP 1';the S stage stands for `synthesis', during which DNA replication occurs; and, the G2 stage stands for `GAP 2'.

A sequence-specific DNA-binding factor, sometimes referred to as TF, is a protein that binds to specific DNA sequences, thereby controlling the flow (or transcription) of genetic information from DNA to mRNA.  In their experiment, \citet{Spellman1998}  identified approximately 800 genes that vary in a periodic fashion during the yeast cell-cycle.  The regulation of most of these genes was not clear; however, TFs have been observed to play critical roles in gene expression regulation \citep[see also][]{Simon2001}.
We apply the proposed method to identify the TFs that influence the gene expression level at each stage of the {\color{black} cell-cycle} process. This is essential for understanding how the cell-cycle is regulated and also how cell-cycles regulate other biological processes.

We analyze a subset of 297 cell-cycle-regularized genes {\color{black} obtained from the \texttt{R} package \texttt{PGEE} \citep{inan2017pgee}} as in \citet{Luan2003cluster} and  \citet{Wang2008repmeasurement}. The response variable $y_{ij}$ is the log-transformed gene expression level of gene $i$, measured at time point $j$; the covariate $\mathbf{x}_{ik}$, $k = 1, \ldots, 96$, is the matching score of the binding probability of the $k$th transcription factor on the promoter region of the $i$th gene.
{\color{black} We calculated the binding probability with a mixture modeling approach based on data from a ChIP binding experiment \citep[see][for further details]{Wang2007gSCAD} and, to compare the performance of ASE-D, ASE-R, and GEE, we removed 20 less correlated covariates due to the sample size available.} 

We apply the sequential estimation procedure to the  G1 stage (which contains a few time points from the cycle) of the cell-cycle process using the following model
\[
  y_{ij} = \alpha_0 + \alpha_1 t_{ij} + \sum_{k=1}^{60}\beta_k \mathbf{x}_{ik} + \epsilon_{ij},
\]
where $\mathbf{x}_{ik}, k=1,\ldots, 60$, is standardized to have mean zero and variance 1; $t_{ij}$ denotes time and  imposes shrinkage on the $\beta_k$'s.
Table \ref{tab:yeast} summarizes the number of TFs identified based on 1000 replications,  using three different working correlation structures for $\epsilon_{ij}$: independence, \textsc{ar}(1), and exchangeable{\color{black}.  In each replication, the initial regression coefficient estimates are based on $100$ random observations.}.
Our analysis reveals that at the G1 stage, the selected TFs, in terms of numbers and specific TFs, are not sensitive to the choice of the working correlation structure. Due to the space limitation, in Table \ref{tab:yeast}, we only report the average sample sizes, the average number of selected TFs and the average estimates of the correlation parameter. Some of the TFs selected have already been confirmed by biological experiments using the genome-wide binding method. For example, MBP1, SWI6, and SWI4 are three TFs that have been proved important in stage G1 in the aforementioned biological experiments and they have been selected by the proposed sequential procedure for stage G1.  We can see from Table \ref{tab:yeast} that ASE-D uses fewer observations ASE-R based on their averages; in fact, the standard deviation of the sample sizes used by ASE-D is also smaller than that of the sample sizes used by ASE-R.  {\color{black} The numerical results here clearly show that the proposed method is beneficial and the sequential procedure with D-optimality for adaptive data recruiting is promising.}

\begin{table}
% \def~{\hphantom{0}}
% \tabcolsep=3pt\fontsize{9}{8}
% \selectfont
\centering
\caption{Sample size, number of TFs selected, and estimates of the correlation parameter in the yeast cell-cycle process with ASE-R and ASE-D.}
    \begin{tabular}{ccccc}%\toprule
    \hline
    Correlation  & Method & $N$        & $\hat{p}_0$ & $\hat{\alpha}$ \\% \midrule%\cmidrule
    \hline
    \textsc{ar}(1)         & ASE-R  & 240(26.20) & 22.7(5.60)      & 0.498(0.021) \\
                 & ASE-D  & 133 (8.37)  & 21.0(5.96)      & 0.379(0.045) \\
    Exchangeable & ASE-R  & 240(26.20) & 21.3(5.63)      & 0.339(0.023) \\
                 & ASE-D  & 133 (7.88)  & 20.2(5.63)      & 0.215(0.050)\\
    Independence & ASE-R  & 231(26.80)  & 15.0(3.31)      & -- \\
                 & ASE-D  & 129 (6.52)  & 17.1(4.32)      & -- \\
    %\bottomrule
    \multicolumn{5}{l}{*  Empirical standard deviations are given in parentheses}
    \end{tabular}%
  \label{tab:yeast}%
\end{table}%

\subsection*{Multiple Sclerosis Data Analysis}

This is a longitudinal clinical trial data set for assessing the effects of neutralising antibodies on interferon beta-1b (IFNB) in relapsing–remitting multiple sclerosis (MS), a disease that destroys the myelin sheath surrounding the nerves.
This data set is from a Magnetic Resonance Imaging (MRI) sub-study of the Betaseron clinical trial conducted at the University of British Columbia in relapsing–remitting multiple sclerosis, which involved 50 patients and each one visited the university every six weeks.
The patients were randomly allocated  into three treatment groups: 17 patients treated by placebo, 17 by a low dose, and 16 by a high dose.
This data set  was previously analyzed in \citet{Petkau2003,petkau2004} and included in a book by \citet{song2007}.

Exacerbation is used as the binary response variable, which indicates whether an exacerbation appeared since the previous MRI scan---1 for `yes' and 0 for `no'.
Seven explanatory variables were recorded: Treatment (Trt), Time (T) in weeks, Squared time (T$^2$), Age, Gender, Duration of disease (Dur) in years, and an additional baseline covariate---initial EDSS (Expanded Disability Status Scale) scores.
 Due to the results in \citet{song2007} and \citet{LI2013csda} and the available sample size,  we also delete Age from the original data as in the analysis conducted in \citet{LI2013csda}.  {\color{black} We recode Ltrt and Htrt as
\[
\text{Ltrt} =
\begin{cases}
  1, & \text{Low Dose}\\
  0, & \text{Otherwise},
\end{cases}
\qquad
\text{Htrt} =
\begin{cases}
  1, & \text{High Dose}\\
  0, & \text{Otherwise}.
\end{cases}
\]
}
and consider the marginal logistic model
\[
\text{logit}(\mu_{ij})=\beta_1\text{T}_j + \beta_2\text{T}_j^2 + \beta_3\text{Gender}_i+
\beta_4\text{Dur}_i + \beta_5\text{EDSS}_i + \beta_6\text{Ltrt}_i + \beta_7\text{Htrt}_i,
\]
where $\mu_{ij}$ is the probability of exacerbation at visit $j$ for subject $i$.  {\color{black} For illustration, we use three correlation structures: \textsc{ar}(1), exchangeable, and independence.}
Table \ref{tab:ms} reports the average sample size ($N$), the number {\color{black} of simulation runs} that use the entire sample ($N^{+}$), the average number of variables selected ($\hat{p}_0$) and the estimate of $\alpha$ ($\widehat\alpha$) based on 1000 runs. {\color{black} We start with an initial estimate of regression coefficients using $25$ randomly selected observations in each run.}
In Table \ref{tab:ms.beta}, for each variable, we also summarize how many times the corresponding coefficient shrinks to zero during the 1000 runs.
Our numerical results show that ASE-D tends to use fewer observations and can complete the estimation procedure with the available sample sizes.
{\color{black} Although the differences may not be statistically significant, they clearly show that ASE-D tends to select fewer variables and this information is usually beneficial for practitioners to design future studies.}

\begin{table}
% \def~{\hphantom{0}}
% \tabcolsep=3pt\fontsize{10}{9}
% \selectfont
\centering
\caption{Sample size, number of simulation runs for which the entire sample has been used, number of variables selected, and correlation parameter estimates in the multiple sclerosis data process with ASE-R and ASE-D.}
    \begin{tabular}{cccccc}%\toprule
    \hline
    Correlation & Method  & $N$ & $N^{+}$ &$\hat{p}_0$ & $\hat{\alpha}$ \\% \midrule%\cmidrule
    \hline
    \textsc{ar}(1)         & ASE-R & 40.8(9.00) & 310 & 3.31(0.493) & -0.080(0.013) \\
                 & ASE-D & 36.2(6.72) & 68  & 3.07(0.253) & -0.072(0.012) \\
    Exchangeable & ASE-R & 42.0(9.67) & 368 & 3.38(0.489) & 0.018(0.013) \\
                 & ASE-D & 37.0(7.50) & 110 & 3.12(0.325) & 0.026(0.014) \\
    Independence & ASE-R & 41.6(9.11) & 360 & 3.55(0.512) & -- \\
                 & ASE-D & 36.2(6.87) & 89  & 3.09(0.290) & -- \\
    %\bottomrule

    \multicolumn{6}{l}{*  Empirical standard deviations are given in parentheses}\\
    \hline
    \end{tabular}%
  \label{tab:ms}%
\end{table}%

\begin{table}
% \tabcolsep=3pt\fontsize{10}{9}
% \selectfont
\centering
  \caption{Number of coefficients that shrunk to zero over 1000 simulations of each variable.}
    \begin{tabular}{lrrrrrr}
    \hline
\multicolumn{1}{l}{
%\multirow{2}[0]
{Coefficients}}
& \multicolumn{2}{c}{\textsc{ar}(1)} & \multicolumn{2}{c}{Exchangeable} & \multicolumn{2}{c}{Independence} \\
& \multicolumn{1}{c}{ASE-R} & \multicolumn{1}{c}{ASE-D}
& \multicolumn{1}{c}{ASE-R} & \multicolumn{1}{c}{ASE-D}
& \multicolumn{1}{c}{ASE-R} & \multicolumn{1}{c}{ASE-D}\\
         %\midrule
        \hline 
T
    & 0 & 0
    & 0 & 0
    & 0 & 0 \\
T$^2$
    & 994 & 998
    & 994 & 997
    & 992 & 996\\
Gender
    & 571 & 900
    & 514 & 845
    & 539 & 867 \\
Durr
    & 1000 & 1000
    & 1000 & 1000
    & 999 & 999\\
EDSS
    & 2 & 2
    & 4 & 3
    & 4 & 1\\
Ltrt
    & 1000   & 923
    & 951    & 994
    & 943    & 997 \\
Htrt
    & 230 & 57
    & 166 & 61
    & 188 & 60 \\
    \hline
    \end{tabular}%
  \label{tab:ms.beta}%
\end{table}%

\section{Discussion}
The collection of correlated or highly stratified response data is common due to modern methods of data collection and to the monitoring methods and facilities.
Data set sizes are larger than before, and, hence, using an entire data set at once may not be convenient when our computational power and skills are not adequate.  Selecting a random subset of such a large data set is a common and easy solution; however, it is well-known that this method usually suffers from sampling variations and cannot always provide a stable and consistent result.  Moreover, to decide the size of such a random subset has never been an easy task when complicated models are involved.
Thus,  how to efficiently use the effective observations is an important issue and the idea of a sequential estimation method seem appropriate for our needs. We propose a method providing sequential stochastic MQLEs so that we can adaptively select effective observations and adopt the ASE method to detect the effective variables simultaneously.
The proposed method recruits new observations from a data pool into the analysis until  a pre-specified stopping criterion is fulfilled.  The asymptotic properties of the sequential procedure and numerical studies using both simulation data and real data examples show that the proposed method can perform well and appears to be a promising method for practical uses.  Our method is not limited to selecting one new observation at a time.  Selecting a batch of observations at each stage is possible with only slight modifications.

\section*{Acknowledgment}
A part of this research was supported by the Ministry of Science and Technology, Taiwan (MOST 106-2118-M-001-007-MY2). The first author also gratefully acknowledges the financial support from the China Scholarship Council(Grant No. 201706340175).

\section*{Supplementary materials}
%\label{SM}
We state the results of the sequential procedure with no ASE feature in  Supplementary materials, together with some extra simulation results, including the cases of the logit link with the ASE feature and different correlation structures. Readers interested in the results under different setups may refer to Supplementary materials.

\appendix
%\appendixone
\section*{Appendix}

\subsection*{Technical details}
\begin{proof}[Proof of Lemma \ref{mqle:ucip}]
%{\noindent Proof of Lemma \ref{mqle:ucip}}:\\
Let $\mathbf{S}_n(\beta_0)=\sum_{i=1}^n\mathbf{X}_i\mathbf{A}_i\mathbf{V}_i^{-1}\mathbf{e}_i$.
Because $\mathbf{X}_i$ is $\mathcal{F}_{i-1}$ measurable,  $\{\mathbf{e}_i:i \geq 1\}$ is a sequence of martingale differences with respect to $\{ \mathcal{F}_i, i \geq 1\}$, and \(\mathbf{S}_n(\beta_0)\) is a sum of martingale differences.
By Taylor series expansion, there exists $\bar\beta_n$, which lies between $\widehat\beta$ and $\beta_0$, such that
\begin{align*}
  \mathbf{M}_n^{-1/2}\mathbf{H}_n(\widehat\beta_n-\beta_0) = &\; \mathbf{M}_n^{-1/2}\mathbf{H}_n^{1/2}\bigl\{\mathbf{H}_n^{-1/2}\mathcal{D}_n(\bar\beta_n) \mathbf{H}_n^{-1/2}\bigr\}^{-1}\\
  & \times \bigl\{ \mathbf{H}_n^{-1/2}\mathbf{M}_n^{1/2}\bigr\}^{-1}\mathbf{M}_n^{-1/2}\mathbf{S}_n(\beta_0),
\end{align*}
where $\mathbf{S}_n(\beta_0)=\mathcal{D}_n(\bar\beta_n)(\widehat\beta_n-\beta_0)$. \citet{yin2006asymptotic} showed that $\mathbf{H}_n^{-1/2}\mathcal{D}_n(\bar\beta_n) \mathbf{H}_n^{-1/2} \rightarrow I_p$ almost surely under conditions (C1)--(C3). Note that each element of $\mathbf{M}_n^{-1/2}\mathbf{H}_n^{-1/2}$ and its inverse matrix are bounded. Hence, by Lemma 1.4 of \citet{Woodroofe1982Nonlinear}, to show $\{ \mathbf{M}_n^{-1/2}\mathbf{H}_n(\widehat{\beta}_n-\beta_0): n \geq 1\}$ is u.c.i.p., it suffices to show that $\{ \mathbf{S}_n(\beta_0): n\geq 1\}$ is u.c.i.p.

By the assumption that $\sup_{i\geq1} \|\mathbf{X}_i\| < \infty$ almost surely, all $\mathbf X_i\tr\beta_0$ fall in a compact set $T$ of $\mathbf{R}^q$ with probability one.  A result of \citet[][ Lemma 1]{yin2006asymptotic} implies that for each $i \geq 1$ and $t \in T$, $c_1I_q<\mathbf{A}_i(t)< c_2I_q$ and
$c_3I_q<\mathbf{V}_i^{-1}(t)< c_4I_q$ with probability one. Let $\lambda_{\max}(\mathbf{T})(\lambda_{\min}(\mathbf{T}))$ denote the largest(smallest) eigenvalue of the matrix $\mathbf{T}$.
Then,
\begin{align*}
 \left\|\mathbf{X}_i\mathbf{A}_i\mathbf{V}_i^{-1}\mathbf{e}_i\right\|^2
 &= \mathrm{tr}\left(\mathbf{X}_i\mathbf{A}_i\mathbf{V}_i^{-1}\Cov(\mathbf{y}_i\mid\mathcal{F}_{i-1})\mathbf{V}_i^{-1}\mathbf{A}_i\tr\ \mathbf{X}_i\right)\\
 &\leq \lambda_{\max}\left(\mathbf{X}_i\mathbf{A}_i\mathbf{V}_i^{-1}\Cov(\mathbf{y}_i\mid\mathcal{F}_{i-1})\mathbf{V}_i^{-1}\mathbf{A}_i\tr\ \mathbf{X}_i\right)
 \leq c.
\end{align*}
Let $U_n=\mathbf{S}_n(\beta_0)$ and $b_n=1$ in the H\'{a}jek-R\'{e}nyi inequality \cite[see][Theorem 7.4.8 (iii)]{chow1988probability}.  Following the arguments of Example 1.8 in \citet{Woodroofe1982Nonlinear} and replacing the Kolmogorov's inequality there with the H\'{a}jek-R\'{e}nyi inequality for martingale differences, we have, for $\epsilon, \delta > 0$ and $k \leq n\delta$,
\begin{align*}
  \mathrm{pr}\left\{\max_{k \leq n\delta} |\mathbf{S}_{n+k}-\mathbf{S}_n| \geq \epsilon\sqrt{n}\right\} \leq  \left(\frac{c}{n\epsilon^2}\right)n\delta=\frac{c\delta}{\epsilon^2}.
\end{align*}
Because ${c\delta}/{\epsilon^2}$ is independent of $n$ and goes to zero as $\delta\rightarrow 0$, this implies that $\{ \mathbf{S}_n(\beta_0), n\geq 1\}$ is u.c.i.p.
\end{proof}

%-------proof
\begin{proof}[Proof of Theorem~\ref{thm:ucip.rho}]
%{\noindent \bf Proof of Theorem~\ref{thm:ucip.rho}}:
%
Let us rewrite
  \begin{align*}%\label{eq:lem1.1} %21
    \sqrt{\rho(n)}(\widehat \beta _n-\beta_0)&=\sqrt{\rho(n)}I_n(\epsilon)(\tilde\beta_n-\beta_0)+\sqrt{\rho(n)}(I_n(\epsilon)-I_0)\beta_0
    \equiv \Delta_1(n)+\Delta_2(n).
  \end{align*}
Because $I_n(\epsilon)$ converges to $I_0$ almost surely, this implies that $\Delta_2(n)$ converges almost surely to $0$
  as $n$ goes to infinity.  It has been shown that  $\mathbf{M}_n^{-1/2}\mathbf{H}_n(\tilde\beta_n-\beta_0) \rightarrow N(0, I_p)$ in distribution \citep{yin2006asymptotic}. Then, by Slusky's theorem and (C6),
$\sqrt\rho(n)(\widehat\beta_{n}-\beta_0)\rightarrow N(0,I_0\vSigma^{-1}I_0)$ in distribution. Because  $I_n(\epsilon)$ is bounded, to replace the sample size with a random variable it is sufficient to show that  $\{\sqrt{\rho(n)}(\tilde \beta_n-\beta_0)\}$ is u.c.i.p. \citep[see][]{anscombe1951large}.
  \begin{align}\label{eq:lem1.1}
    \sqrt{\rho(n)}(\tilde \beta_n-\beta_0)=&\left\{\left(\frac{\mathbf{M}_n^{-1/2}\mathbf{H}_n}{\sqrt{\rho(n)}}\right)^{-1}-\vSigma^{-1/2}\right\}\sqrt{\rho(n)}\mathbf{M}_n^{-1/2}\mathbf{H}_n(\tilde \beta_n-\beta_0)\nonumber\\
    &+\sqrt{\rho(n)}\vSigma^{-1/2}\mathbf{M}_n^{-1/2}\mathbf{H}_n(\tilde \beta_n-\beta_0).
  \end{align}

Because $\sqrt{\rho(n)}\mathbf{M}_n^{-1/2}\mathbf{H}_n(\tilde \beta_n-\beta_0)$ is stochastic bounded, it follows from condition (C6) that  the first term of \eqref{eq:lem1.1} goes  to $0$ almost surely as $n$ goes to infinity. Using arguments similar to those in  Lemma~\ref{mqle:ucip}, the sequence $\{\sqrt{\rho(n)}\vSigma^{-1/2}\mathbf{M}_n^{-1/2}\mathbf{H}_n(\tilde \beta_n-\beta_0)\}$ is u.c.i.p, and the remain proof, which is omitted here, follows from the results of~\citet{anscombe1951large} (see also~\cite{Woodroofe1982Nonlinear}).
\end{proof}

%-------proof
\begin{proof}[Proof of Theorem~\ref{thm:shrinkage.seq}]
%{\noindent \bf Proof of Theorem~\ref{thm:shrinkage.seq}}
{
We know that with probability one, \\
$\rho(k)I_k(\epsilon)(\widehat{\mathbf{H}}_k\widehat{\mathbf{M}}_k^{-1}\widehat{\mathbf{H}}_k)^{-}I_k(\epsilon)$ converges to $I_0\vSigma^{-1}I_0$ and $\rho(k)\nu_k$ converges to $\nu$  as $k \rightarrow \infty$.
Let $y_k={(\rho(k)\nu_k)}/{\nu}$, $f(k)={(\rho(k)a^2)}/{a_n^2}$ and $t = {(\nu a^2)}/{d^2}$.
Then, \eqref{eq:stoprule.noshrinkage} becomes
$
N=\min\{k: k\geq n\ \text{and}\ y_k\leq f(k)/t\}.
$
Hence,   $1=\lim_{t\rightarrow \infty}{f(N)}/{t}=\lim_{d\rightarrow
    0}{(d^2\rho(N)}){(a^2\nu)}$ {almost surely.}
This implies that,  as $t \rightarrow \infty$, ${d^2\rho(N)}/{\nu_N} \rightarrow a^2$ and $\rho(N)/t\rightarrow 1$ almost surely.
Hence, (i) follows from \citet[][Lemma 1]{chow1965asymptotic}.

Let $\mathcal M_n=\widehat{\mathbf{H}}_n\widehat{\mathbf{M}}_n^{-1}\widehat{\mathbf{H}}_n$.
Hence, when $t$ goes to infinity, the event $\{\mathbf{\beta}_0 \in R_N\}$  is equivalent to %\big)
$$\biggl\{\rho(N)(\widehat\beta_N-\beta_0)\tr \{I_N(\epsilon)   {\mathcal M_N}^{-1}I_N(\epsilon)\}^{-}(\widehat\beta_N-\beta_0)\leq \frac{d^2\rho(N)}{\nu_N}
\text{ and } \beta_j = 0 \text{ for } I_{Nj}(\epsilon) = 0
\biggr\}$$
Thus, Theorem~\ref{thm:shrinkage.seq} (ii) follows from Theorem~\ref{thm:ucip.rho}.
For the proofs of (iii) and (iv), we can  use arguments similar to those  of \citet[][Theorem 4]{wang2013sequential}, so we are omitting them here.
}

\end{proof}

% Submissions are not required to reflect the precise reference formatting of the journal (use of italics, bold etc.), however it is important that all key elements of each reference are included.

\bibliographystyle{imsart-nameyear}
%\bibliography{Bib}

% \includepdf[pages={1-32}]{JRSSB-SeqGee-Supp}

\end{document}